\documentclass[11pt]{article}

\usepackage{amssymb}
\usepackage{amsmath}
\usepackage{amsthm}
\usepackage{fullpage}
\usepackage{graphicx}

\newtheorem{proposition}{Proposition}
\newtheorem{lemma}{Lemma}
\newtheorem{corollary}{Corollary}

\usepackage{hyperref}
\usepackage{overpic}
\usepackage{color}
\usepackage{bbm}  
\usepackage{overpic}
\usepackage[numbers]{natbib}

\newcommand{\usedim}{d}
\newcommand{\mc}[1]{\mathcal{#1}}
\newcommand{\card}{\mathop{\rm card}}
\renewcommand{\P}{\mathbb{P}}
\newcommand{\E}{\mathbb{E}}
\newcommand{\R}{\mathbb{R}}
\newcommand{\Z}{\mathbb{Z}}

\newcommand{\defeq}{:=}
\newcommand{\what}[1]{\widehat{#1}}
\newcommand{\Ball}{\mathbb{B}}
\newcommand{\ball}{\Ball}

\newcommand{\floor}[1]{\left\lfloor{#1}\right\rfloor}
\newcommand{\ceil}[1]{\left\lceil{#1}\right\rceil}

\newcommand{\cov}{\mathop{\rm Cov}}

\newcommand{\metric}{\rho}
\newcommand{\hoodsize}{N}
\newcommand{\hoodbig}{\hoodsize^{\max}}
\newcommand{\hoodsmall}{\hoodsize^{\min}}
\newcommand{\packrv}{V}
\newcommand{\packval}{v}
\newcommand{\altpackval}{w}
\newcommand{\packset}{\mc{V}}
\newcommand{\packmetric}{\metric_{\packset}}

\newcommand{\minimax}{\mathfrak{M}}
\newcommand{\indic}[1]{\mathbbm{1}\left\{#1\right\}}
\newcommand{\argmin}{\mathop{\rm argmin}}

\newcommand{\dkl}[2]{D_{\rm kl}\left({#1} |\!| {#2}\right)}
\newcommand{\hinge}[1]{\left[{#1}\right]_+}

\newcommand{\norm}[1]{\left\|{#1}\right\|}
\newcommand{\norms}[1]{\|{#1}\|}
\newcommand{\ltwo}[1]{\norm{#1}_2}
\newcommand{\lone}[1]{\norm{#1}_1}
\newcommand{\ltwos}[1]{\norms{#1}_2}

\newcommand{\lfro}[1]{\norm{#1}_{\rm Fr}}
\newcommand{\tr}{\mathop{\rm tr}}

\newcommand{\normal}{\mathsf{N}}
\newcommand{\stddev}{\sigma}

\newcommand{\half}{\frac{1}{2}}

\newcommand{\Vhat}{\ensuremath{\what{V}}}
\newcommand{\Exs}{\ensuremath{\mathbb{E}}}

\makeatletter
\long\def\@makecaption#1#2{
  \vskip 0.8ex
  \setbox\@tempboxa\hbox{\small {\bf #1:} #2}
  \parindent 1.5em  
  \dimen0=\hsize
  \advance\dimen0 by -3em
  \ifdim \wd\@tempboxa >\dimen0
  \hbox to \hsize{
    \parindent 0em
    \hfil
    \parbox{\dimen0}{\def\baselinestretch{0.96}\small
      {\bf #1.} #2
    }
    \hfil}
  \else \hbox to \hsize{\hfil \box\@tempboxa \hfil}
  \fi
}
\makeatother

\author{John C. Duchi ~~~~ Martin J. Wainwright\\
\texttt{\{jduchi,wainwrig\}@berkeley.edu}}

\begin{document}

\begin{center}
  {\Large{\bf{ Distance-based and continuum Fano inequalities \\ with
        applications to statistical estimation}}}

  \vspace*{.35cm}
  
  \begin{tabular}{ccc}
    John C.\ Duchi$^\dagger$ && Martin J.\ Wainwright$^{\dagger, \ast}$
    \\ \texttt{jduchi@eecs.berkeley.edu} &&
    \texttt{wainwrig@berkeley.edu} \\
  \end{tabular}
  
  \vspace*{.35cm}
  
  \begin{tabular}{ccc}
    Department of Electrical
    Engineering and Computer Science$^{\dagger}$
    && Department of Statistics$^\ast$
  \end{tabular}

  \vspace*{.2cm}
  
  \begin{tabular}{c}
    University of California, Berkeley, \\ Berkeley, CA, 94720
  \end{tabular}
\end{center}
\vspace{.1cm}

\begin{abstract}
  In this technical note, we give two extensions of the classical Fano
  inequality in information theory. The first extends Fano's
  inequality to the setting of estimation, providing lower bounds on
  the probability that an estimator of a discrete quantity is within
  some distance $t$ of the quantity. The second inequality extends our
  bound to a continuum setting and provides a volume-based bound. We
  illustrate how these inequalities lead to direct and simple proofs
  of several statistical minimax lower bounds.
\end{abstract}


\section{Introduction}

Fano's inequality is a central tool in information theory, serving as
a key ingredient in the proofs of many converse
results~\cite{Fano52,CoverTh06}. In mathematical statistics, it also
plays a central role in minimax theory---specifically, in proving
lower bounds on achievable rates of convergence for estimators.  This
line of work in statistics dates back to the seminal work of
Hasminskii and Ibragimov~\cite{Hasminskii78,IbragimovHa81}, and
continues through a variety of works to the present day
(e.g.,~\cite{Birge83,Birge05,RaskuttiWaYu09, SanthanamWa12, YangBa99,
  Yu97}).

Let us begin by stating Fano's inequality, which provides a lower
bound on the error in a multi-way hypothesis testing problem.  Let $V$
be a random variable taking values in a finite set $\packset$ with
cardinality $|\packset| \geq 2$.  If we define the binary entropy
function $h_2(p) = -p \log p - (1 - p) \log (1 - p)$, Fano's
inequality takes the following form~\cite[e.g.][Chapter 2]{CoverTh06}:
\begin{lemma}[Fano]
  \label{lemma:fano}
  For any Markov chain $V \rightarrow X \rightarrow \Vhat$, we have
  \begin{equation}
    \label{eqn:fano}
    h_2(\P(\what{V} \neq V)) + \P(\what{V} \neq V) \log(|\packset| - 1)
    \ge H(V \mid \what{V}).
  \end{equation}
\end{lemma}
\noindent
A standard simplification of Lemma~\ref{lemma:fano} is to note that
$h_2(p) \le \log 2$ for any $p \in [0, 1]$, so that if $V$ is uniform
on the set $\packset$ and hence $H(V) = \log |\packset|$, then 
\begin{equation}
  \label{eqn:fano-information}
  \P(\what{V} \neq V) \ge 1 - \frac{I(V; X) + \log 2}{\log
    |\packset|}.
\end{equation}
Note that the probability $\P(\Vhat \neq V)$ can be interpreted as the
error in a $|\packset|$-ary hypothesis testing problem, where $X$
represents the observation, $V$ represents the latent class label,
and $\Vhat$ represents a testing function.

When applied to derive statistical minimax bounds on a estimation
problem, the usual procedure is to ``reduce'' the estimation problem
to a testing problem before applying the usual form of Fano's
inequality.  See the papers of~\citet{Yu97} and~\citet{YangBa99} for a
description of this standard reduction.  In this note, we provide
extensions of inequalities~\eqref{eqn:fano}
and~\eqref{eqn:fano-information} that directly yield bounds on the
estimation error, thereby allowing this step to be avoided.  More
specifically, suppose that we have the distance-like function $\metric$ on
$\packset$ and are interested in bounding the estimation error
$\metric(\Vhat, V)$.  We begin by providing analogues of the lower
bounds~\eqref{eqn:fano} and~\eqref{eqn:fano-information} that replace
the testing error with the tail probability $\P(\metric(\what{V}, V) >
t)$.  By Markov's inequality, such control directly yields bounds on
the expectation $\Exs[\metric(\Vhat, V)]$.  We then extend these
bounds to continuous spaces, providing a natural volume-based analogue
of inequality~\eqref{eqn:fano-information}.  As we show in
Section~\ref{sec:consequences}, these distance-based Fano inequalities
allow direct and simple proofs of various minimax bounds without the
need for computing metric entropy as in standard arguments.


\section{Two distance-based Fano inequalities}

In this section, we provide the two main results of this
note---namely, Proposition~\ref{proposition:fun-fano}, which applies
to estimation in discrete problems, and
Proposition~\ref{proposition:continuous-fano} for continuous-valued
variables.

\subsection{Discrete problems}

We begin with the distance-based analogue of the usual discrete Fano
inequality in Lemma~\ref{lemma:fano}.  Let $V$ be a random variable
supported on a finite set $\packset$ with cardinality $|\packset| \geq
2$, and let $\rho : \packset \times \packset \to \R$ be a symmetric
function defined on $\packset \times \packset$.  In the usual setting,
the function $\rho$ is a metric on the space $\packset$, but our
theory applies to general functions.  For a given scalar $t \geq 0$,
the maximum and minimum \emph{neighborhood sizes at radius $t$} are
given by
\begin{equation*}
  \hoodbig_t \defeq \max_{v \in \packset}\left\{ \card \{ v' \in \packset \mid
  \rho(v, v') \le t \}\right\}
  ~~~ \mbox{and} ~~~
  \hoodsmall_t \defeq \min_{v \in \packset}
  \left\{ \card \{ v' \in \packset \mid \rho(v, v') \le t \}\right\}.
\end{equation*}
Defining the error probability $P_t = \P(\metric(\Vhat, V) > t)$, we
then have the following generalization of Fano's inequality:
\begin{proposition}
  \label{proposition:fun-fano}
  For any Markov chain $V \to X \to \what{V}$, we have
  \begin{align}
    \label{eqn:pre-info-processing-fano}
    h_2(P_t) + P_t \log \frac{|\packset| - \hoodsmall_t}{\hoodbig_t} +
    \log \hoodbig_t & \geq H(V \mid \what{V}).
  \end{align}
\end{proposition}

Before proving the proposition, it is informative to
note that it reduces to the standard form of Fano's
inequality~\eqref{eqn:fano} in a special case.  Suppose that we take $\metric$
to be the 0-1 metric, meaning that $\metric(v, v') = 0$ if $v = v'$ and $1$
otherwise.  Setting $t = 0$ in Proposition~\ref{proposition:fun-fano}, we have
\mbox{$P_0 = \P[\Vhat \neq V]$} and \mbox{$\hoodsmall_0 = \hoodbig_0 = 1$,}
whence inequality~\eqref{eqn:pre-info-processing-fano} reduces to
inequality~\eqref{eqn:fano}. Other weakenings allow somewhat clearer
statements:
\begin{corollary}
  \label{corollary:fun-fano}
  If $\packrv$ is uniform on $\packset$ and $(|\packset| - \hoodsmall_t)
  > \hoodbig_t$, then
  \begin{equation}
    \P(\metric(\what{V}, V) > t) \ge 1 - \frac{I(V; X) + \log 2}{\log
      \frac{|\packset|}{\hoodbig_t}}.
    \label{eqn:fun-fano}
  \end{equation}
\end{corollary}
Inequality~\eqref{eqn:fun-fano} is the natural analogue of the
classical mutual-information based form of Fano's
inequality~\eqref{eqn:fano-information}, and it provides a
qualitatively similar bound.  The main difference is that the usual
cardinality $|\packset|$ is replaced by the ratio $|\packset| /
\hoodbig_t$.  This quantity serves as a rough measure of the number of
possible ``regions'' in the space $\packset$ that are
distinguishable---that is, the number of subsets of $\packset$ for
which $\metric(v, v') > t$ when $v$ and $v'$ belong to different
regions.  Such sets are known as packing sets, and their construction
is a standard step in the usual reduction from testing to
estimation~\cite{Yu97,YangBa99}.  Our bound allows us to skip the
packing set construction and directly compute $I(V; X)$ where $V$
takes values over the full space, as opposed to computing the mutual
information $I(V'; X)$ for a $V'$ uniformly distributed over a packing
set contained within $\packset$.  In some cases, the former
calculation can be simpler, as illustrated in examples to follow. We
note that inequality~\eqref{eqn:fun-fano} is similar in spirit to an
inequality of \citet[Lemma IV-1]{AeronSaZh10}, who show an
inequality that asymptotically bounds the error of
tensorized distortion measures in terms of a rate distortion
function. In contrast, our bound is non-asymptotic, and it applies to
general functions $\metric$ and sets $\packset$ without quantization.
\\

\noindent We now prove the corollary:
\begin{proof}
  First, by the information-processing inequality~\cite[e.g.][Chapter
    2]{CoverTh06}, we have \mbox{$I(V; \what{V}) \le I(V; X)$,} and
  hence $H(V \mid X) \le H(V \mid \Vhat)$.  Since $h_2(P_t) \le \log
  2$, inequality~\eqref{eqn:pre-info-processing-fano} implies that
  \begin{equation*} H(V \mid X) - \log \hoodbig_t \le H(V \mid
    \what{V}) - \log \hoodbig_t \le \P(\rho(\what{V}, V) > t) \log
    \frac{|\packset| - \hoodsmall_t}{\hoodbig_t} + \log 2.
  \end{equation*}
  Rearranging the preceding equations yields 
  \begin{align}
    \label{EqnCaffeTrieste}
    \P(\rho(\what{V}, V) > t) & \geq \frac{H(V \mid X) - \log \hoodbig_t -
      \log 2}{ \log \frac{|\packset| - \hoodsmall_t}{ \hoodbig_t}}.
  \end{align}
  Note that his bound holds without any assumptions on the distribution
  of $V$.

  By definition, we have $I(V; X) = H(V) - H(V \mid X)$.  When $V$ is
  uniform on $\packset$, we have $H(V) = \log |\packset|$, and hence $H(V
  \mid X) = \log |\packset| - I(V; X)$.  Substituting this relation into
  the bound~\eqref{EqnCaffeTrieste} yields the inequality
  \begin{equation*}
    \P(\rho(\what{V}, V) > t) \ge \frac{\log
      \frac{|\packset|}{\hoodbig_t}}{ \log \frac{|\packset| -
        \hoodsmall_t}{\hoodbig_t}} - \frac{I(V; X) + \log 2}{
      \log\frac{|\packset| - \hoodsmall_t}{\hoodbig_t}} \ge 1 -
    \frac{I(V; X) + \log 2}{\log \frac{|\packset|}{\hoodbig_t}},
  \end{equation*}
  as claimed.
\end{proof}

\vspace{.2cm}

\noindent We complete this subsection by proving
Proposition~\ref{proposition:fun-fano}, using an argument that
parallels that of the classical Fano inequality~\cite{CoverTh06}.



\newcommand{\Pvar}{Z}
\begin{proof}
Letting $\Pvar$ be a $\{0,1\}$-valued indicator variable for the event
$\rho(\Vhat, V) \leq t$, we compute the entropy $H(\Pvar, V \mid
\what{V})$ in two different ways.  On one hand, by the chain rule for
entropy, we have
\begin{equation}
\label{eqn:entropy-1}
 H(\Pvar, V \mid \what{V}) = H(V \mid \what{V}) + \underbrace{H(\Pvar
   \mid V, \what{V})}_{= 0},
\end{equation} 
where the final term vanishes since $\Pvar$ is
$(V,\what{V})$-measurable.  On the other hand, we also have
\begin{align} 
H(\Pvar, V \mid \what{V}) & = H(\Pvar \mid \what{V}) + H(V \mid \Pvar,
\what{V}) \leq H(\Pvar) + H(V \mid \Pvar, \what{V}),
\end{align}
using the fact that conditioning reduces entropy.  Applying the
definition of conditional entropy yields
\begin{align*}   
H(V \mid \Pvar, \what{V}) & = \P(\Pvar = 0) H(V \mid \Pvar = 0,
\what{V}) + \P(\Pvar = 1) H(V \mid \Pvar = 1, \what{V}),
\end{align*}
and we upper bound each of these terms separately.  For the first term,
we have
\begin{align*}
 H(V \mid \Pvar = 0, \what{V}) & \leq \log(|\packset| - \hoodsmall_t),
\end{align*}
since conditioned on the event $\Pvar = 0$, the random variable $V$
may take values in a set of size at most $|\packset| - \hoodsmall_t$.
For the second, we have
\begin{align*}
H(V \mid \Pvar = 1, \what{V}) & \leq \log \hoodbig_t,
\end{align*}
since conditioned on $\Pvar = 1$, or equivalently on the event that
$\rho(\Vhat, V) \leq t$, we are guaranteed that $V$ belongs to a set of
cardinality at most $\hoodbig_t$.

Combining the pieces and and noting $\P(\Pvar = 0) = P_t$, we have
proved that
\begin{align*}
  H(\Pvar, V \mid \what{V}) & \leq H(\Pvar) + P_t \log \big(|\packset| -
  \hoodsmall \big) + (1 - P_t) \log \hoodbig_t.
\end{align*}
Combining this inequality with our earlier
equality~\eqref{eqn:entropy-1}, we see that
\begin{equation*}
H(V \mid \what{V}) \le H(\Pvar) + P_t\log (|\packset| - \hoodsmall_t)
+ (1 - P_t) \log \hoodbig_t.
\end{equation*}
Since $H(\Pvar) = h_2(P_t)$, the
claim~\eqref{eqn:pre-info-processing-fano} follows.
\end{proof}


\subsection{Continuous problems}
\newcommand{\vol}{\mathop{\rm Vol}}

Thus far, we have considered problems in which the random variable $V$ takes
values in a discrete set of finite cardinality.  In this section, we show how
to extend Proposition~\ref{proposition:fun-fano}---specifically, its
mutual-information based form~\eqref{eqn:fun-fano}---to non-discrete
domains. This extension has applications to many problems in statistical
decision theory, as sketched in Section~\ref{sec:apply-continuous-fano},
allowing quick proofs of several results. It may prove useful in other domains
as well, though for brevity we focus only on a few examples.

The proposition requires a bit of additional notation and a few
assumptions.  First, we assume that the set $\packset \subset \R^d$
has a volume (Lebesgue measure) that is non-zero and finite.  Define
the ``ball''
\begin{align*}
\ball_\metric(t, v)
= \{v' \in \R^d : \metric(v, v') \le t\}
\end{align*}
of ``radius'' $t$ centered at $v$. Here we use ball and radius with
quotation marks, since $\metric$ may not be a metric, let alone
symmetric or even positive.  For a set $\mathcal{S}$ in
$d$-dimensions, we let $\partial \mathcal{S}$ denote its boundary, and
define the volume of $\partial \mathcal{S}$ via Lebesgue measure in
$(d-1)$-dimensional space.  With this notation, we assume the volumes
of the two surface areas $\vol(\partial \packset)$ and $\sup_{v \in
  \packset} \vol(\partial (\ball_\metric(t, v) \cap \packset))$ are
both finite.  Under this regularity condition
on the pair $(\metric, \packset)$, we have the following result:
\begin{proposition}
  \label{proposition:continuous-fano}
If $V$ is uniform over $\packset$, then for any Markov chain $V \to X
\to \what{V}$, we have
\begin{equation}
\P(\metric(\what{V}, V) \ge t) \ge 1 - \frac{I(V; X) + \log 2}{\log
  \frac{\vol(\packset)}{ \sup_{v \in \packset} \vol(\ball_\metric(t,
    v) \cap \packset)}}.
 \label{eqn:continuous-fano}
\end{equation}
\end{proposition}
\noindent See Appendix~\ref{sec:proof-continuous-fano} for the proof.

Given the generalized Fano inequality~\eqref{eqn:fun-fano}, the form of
Proposition~\ref{proposition:continuous-fano} is not surprising: the volume
ratio in the bound~\eqref{eqn:continuous-fano} is the continuous analog of the
ratio $\frac{|\packset|}{\hoodbig_t}$ in the discrete case.  The
bound~\eqref{eqn:continuous-fano} is related to some recent results
by~\citet{MaWu13}.  In their work, the volume ratio is introduced as a lower
bound on packing entropy, following the usual reduction from estimation to
testing.  Consequently, unlike the bound~\eqref{eqn:continuous-fano}, it does
not allow the random vector $V$ to take continuous values.  The main advantage
of Proposition~\ref{proposition:continuous-fano} is that it provides direct
bounds on estimation error for continuous random vectors $\packrv$ in terms of
mutual information, without the need for performing discretization and
bounding metric entropy.


\section{Consequences for statistical minimax theory}
\label{sec:consequences}

In this section, we develop some consequences of the previous results
for statistical minimax theory.  Accordingly, we begin by setting up
the standard framework for minimax bounds in statistics.  Let $\mc{P}$
be a family of distributions on a sample space $\mc{X}$, and let
$\theta : \mc{P} \to \Theta$ be a function mapping $\mc{P}$ to some
parameter space $\Theta$. Given a sample $X_1^n = (X_1, \ldots, X_n)
\in \mc{X}^n$ of size $n$ drawn i.i.d.\ from a distribution $P \in
\mc{P}$, let $\what{\theta} \equiv \what{\theta}(X_1^n)$ be a
measureable function of $X_1^n$, which we view as an estimate of the
unknown quantity $\theta(P)$.  The quality of this estimator can be
measured in terms of the \emph{risk}
\begin{align}
  \label{EqnMSE}
  \E_P\left[\Phi\left(\metric(\what{\theta}(X_1^n), \theta(P))\right)\right],
\end{align}
where $\metric : \Theta \times \Theta \to \R_+$ is a (semi)-metric on
the parameter space, and $\Phi : \R_+ \to \R_+$ is a non-decreasing loss
function.  For instance, for Euclidean space with
$\metric(\what{\theta}, \theta) = \|\what{\theta} - \theta\|_2$ and
$\Phi(t) = t^2$, the error measure corresponds to the usual
mean-squared error of an estimator.  In terms of this
criterion~\eqref{EqnMSE}, the \emph{minimax risk} for the family
$\mc{P}$ is given by
\begin{equation}
  \label{eqn:minimax-risk}
  \minimax_n(\theta(\mc{P}), \Phi \circ \metric) \defeq
  \inf_{\what{\theta}} \sup_{P \in \mc{P}}
  \E_P\left[\Phi\left(\metric(\what{\theta}(X_1^n), \theta(P))\right)\right],
\end{equation}
where the infimum ranges over all measureable functions
$\what{\theta}$ of the observed sample $X_1^n$.

\subsection{Consequences of Proposition~\ref{proposition:fun-fano}}

We begin by showing how Proposition~\ref{proposition:fun-fano} can be
used to derive lower bounds on the minimax risk.  As previously noted,
Proposition~\ref{proposition:fun-fano} is a generalization of the
classical Fano inequality~\eqref{eqn:fano}, so it leads naturally
to a generalization of the classical Fano lower bound on minimax
error, which we describe here.

Consider a family of distributions $\{P(\cdot \mid
\packval)\}_{\packval \in \packset} \subset \mc{P}$ indexed by a
finite set $\packset$.  This family induces an associated collection of
parameters $\{\theta_\packval \defeq \theta(P(\cdot \mid
\packval))\}_{\packval \in \packset} \subset \Theta$.  Given a
function $\packmetric : \packset \times \packset \to \R$ and a scalar
$t$, we define the separation $\delta(t)$ of this set relative
to the metric $\metric$ on $\Theta$ via
\begin{equation}
  \label{eqn:separation-func}
  \delta(t) \defeq \sup \left\{\delta \mid \metric(\theta_\packval,
  \theta_{\altpackval}) \ge \delta ~ \mbox{for~all~} \packval,
  \altpackval \in \packset ~ \mbox{such~that} ~ \packmetric(\packval,
  \altpackval) > t \right\}.
\end{equation}
As a special case, when $t = 0$ and $\packmetric$ is the discrete
metric, this definition reduces to that of a packing set: we are
guaranteed that $\metric(\theta_\packval, \theta_\altpackval) \geq
\delta(0)$ for all distinct pairs $\packval \neq \altpackval$.  This
type of packing construction underlies the classical Fano approach to
minimax lower bounds~\cite{IbragimovHa81,Birge83,Yu97,YangBa99}.  On
the other hand, allowing for $t > 0$ lends greater flexibility to the
construction, since only certain pairs $\theta_\packval$ and
$\theta_\altpackval$ are required to be well-separated.  

Given a set $\packset$ and associated separation
function~\eqref{eqn:separation-func}, we assume the canonical
estimation setting: nature chooses a vector $V \in \packset$ uniformly
at random, and conditioned on this choice $V = \packval$, a sample
$X_1^n$ of size $n$ is drawn i.i.d.\ from the distribution $P(\cdot
\mid \packval)$.  We then have the following corollary of
Proposition~\ref{proposition:fun-fano}:
\begin{corollary}
  \label{corollary:estimation-to-testing}
  Given $V$ uniformly distributed over $\packset$ with separation
  function $\delta(t)$, we have
  \begin{align}
    \label{EqnGeneralizedFano}
    \minimax_n(\theta(\mc{P}), \Phi \circ \rho) & \geq \Phi \Big (
    \frac{\delta(t)}{2} \Big) \; \bigg(1 - \frac{I(X_1^n; V) + \log 2}{
      \log \frac{|\packset|}{\hoodbig_t}}\bigg) \qquad \mbox{for all $t$.}
  \end{align}
\end{corollary}
\noindent See Appendix~\ref{sec:proof-estimation-to-test} for the
proof.

The classical form of the Fano lower bound on the
minimax risk can be recovered as a special case of the lower
bound~\eqref{EqnGeneralizedFano}.  Indeed, if we set $t = 0$ and
let $\packmetric$ be the discrete metric, then
$\hoodbig_0 = 1$ and $|\packset|$ is the cardinality of an $\epsilon
\defeq \delta(0)$ packing set, which we denote by $M(\epsilon)$.
Consequently, we obtain
\begin{align}
  \label{eqn:classical-fano}
  \minimax_n(\theta(\mc{P}), \Phi \circ \rho) & \geq \Phi \Big (
  \frac{\epsilon}{2} \Big) \; \bigg(1 - \frac{I(X_1^n; V) + \log 2}{
    \log M(\epsilon)} \bigg),
\end{align}
which is a well-known lower bound on the minimax
risk~\cite{IbragimovHa81,Yu97,YangBa99}.  In general, however,
Corollary~\ref{corollary:estimation-to-testing} gives somewhat more
flexibility than the classical Fano risk
bound~\eqref{eqn:classical-fano}, because it allows $V$ to be
uniformly distributed on a set $\packset$ that may not induce a
well-separated packing (i.e.\ $\epsilon = \delta(0)$ may be small,
weakening the classical bound~\eqref{eqn:classical-fano}), and it
allows more careful choice for the separation $t$. We now give two
illustrative applications of
Corollary~\ref{corollary:estimation-to-testing} that highlight these
aspects of the bound. \\

We first show how Corollary~\ref{corollary:estimation-to-testing} can
be exploited to obtain sharp lower bounds for the problem of sparse
mean estimation.  Suppose that our goal is to estimate the mean
$\theta$ of a Gaussian distribution on $\R^d$, where $\theta \in \R^d$
has at most $s$ non-zero entries (written as $\|\theta\|_0 \leq s$).
We consider the minimax risk in squared Euclidean norm over the family
\begin{align}
  \label{EqnSparseGaussLoc}
  \mc{P}_{s, d} & \defeq \big \{ \normal(\theta, \sigma^2 I_{d \times d}) \,
  \mid \, \norm{\theta}_0 \leq s \big \},
\end{align}
where we observe $X_i$ drawn i.i.d.\ as $\normal(\theta, \sigma^2 I)$.
For this family, we have the following lower bound:
\begin{corollary}
  \label{corollary:sparse-gaussian-location}
  For the $s$-sparse Gaussian location family~\eqref{EqnSparseGaussLoc},
  there is a universal constant $c > 0$ such that
  \begin{align}
    \label{EqnSparseGaussLower}
    \minimax_n \big(\theta(\mc{P}_{s,d}), \; \ltwo{\,\cdot\,}^2 \big) &
    \geq c \, \frac{\sigma^2 s \log (\frac{d}{s})}{n}.
  \end{align}
\end{corollary}
This lower bound is sharp up to a constant factor
(cf.~\citet{DonohoJo94}).  One way of proving the lower
bound~\eqref{EqnSparseGaussLower} is by first constructing a
$\delta$-packing of the set of $s$-sparse mean vectors, and then
applying the classical Fano bound~\eqref{eqn:classical-fano}.  See the
paper~\cite{RaskuttiWaYu09} for an instance of this proof in the
setting of fixed design regression, which includes the normal location
model as a particular
case. Corollary~\ref{corollary:estimation-to-testing}, however, allows
a proof without such a packing construction.
\begin{proof}
  Consider the set $\packset = \{ \packval \in \{-1, 0, 1\}^d \, \mid \,
  \norm{\packval}_0 = s \}$, which satisfies $|\packset| = 2^s \binom{d}{s}$.
  If we define $\theta_\packval = \epsilon \packval$ for some
  $\epsilon > 0$, then the separation function is lower bounded as
  $\delta(t) > \max\{\sqrt{t}, 1\} \: \epsilon$.  Consequently, for
  $\packrv$ uniformly distributed on $\packset$,
  Corollary~\ref{corollary:estimation-to-testing} implies the lower
  bound
  \begin{align*}
    \minimax_n(\theta(\mc{P}_{s,d}), \ltwo{\cdot}^2) > \frac{ (t \vee 1) \;
      \epsilon^2 }{4} \left(1 - \frac{I(\packrv; X_1^n) + \log 2}{\log
      \binom{d}{s} - \log \hoodbig_t}\right).
  \end{align*}
  By scaling it is no loss of generality to assume that $\stddev^2 = 1$.
  For $V$ uniform on $\packset$, we have $\E[\ltwo{\packrv}^2] = s$, and
  thus for $V, W$ independent and uniform on $\mc{V}$,
  \begin{equation*}
    I(\packrv; X) \le n \frac{1}{|\packset|^2}\sum_{\packval \in \packset}
    \sum_{\altpackval \in \packset} \dkl{\normal(\epsilon \packval, I)}{
      \normal(\epsilon \altpackval, I)} = \frac{n \epsilon^2}{2}
    \E\left[\ltwo{\packrv - W}^2\right] = n s \epsilon^2.
  \end{equation*}
  Taking $t = \floor{s/4}$, we find that $\hoodbig_t \le \ceil{s/4}
  2^{\floor{s/4}}\binom{d}{\floor{s/4}}$ and
  \begin{equation*}
    \log \frac{|\packset|}{\hoodbig_t} \ge \log 2^s \binom{d}{s} - \log
    \ceil{s/4} 2^{\floor{s/4}} \binom{d}{\floor{s/4}} \ge \log
    \frac{\floor{s/4}! (d - \floor{s/4})!}{ s! (d - s)!}  \ge c s \log
    \frac{d}{s}
  \end{equation*}
  for a numerical constant $c$. Combining the pieces yields the bound
  \begin{equation*}
    \minimax_n(\theta(\mc{P}_{s,d}),
    \ltwo{\cdot}^2) > \frac{ (\floor{s/4} \vee
      1) \; \epsilon^2}{4} \left(1 - \frac{n \epsilon^2 s + \log 2}{c s
      \log \frac{d}{s}}\right),
  \end{equation*} 
  and setting $\epsilon^2 \asymp \log \frac{d}{s} / n$ yields the
  claim~\eqref{EqnSparseGaussLower}.
\end{proof}

In other applications, the covariance structure of a random vector $V$
uniformly distributed over $\packset$ plays an important role in
giving sharp minimax lower bounds.  Some examples include recent work
on privacy preserving statistical analysis~\cite{DuchiJoWa13_focs},
results obtaining sharp lower bounds for compressive
sensing~\cite[Chapter 6]{CandesDa13,Price13}, and in distributed
estimation problems~\cite{ZhangDuJoWa13_nips}.  In such settings, the
classical Fano minimax bound~\eqref{eqn:classical-fano} requires
rather delicate constructions of packing sets (requiring probabilistic
arguments using matrix concentration, among other techniques);
Corollary~\ref{corollary:estimation-to-testing} allows this
complication to be side-stepped. Here we illustrate this use of
Corollary~\ref{corollary:estimation-to-testing} for the problem of
estimating a sparse vector from linear measurements.

More precisely, suppose that we observe the pair $(Y, X) \in \R^n
\times \R^{n \times d}$ linked via the linear regression model $Y =
X\theta + \varepsilon$, where the noise vector $\varepsilon \in \R^n$
has i.i.d.\ $\normal(0, \stddev^2)$ entries.  We assume that $\theta$
is an $s$-sparse vector, meaning that $\norm{\theta}_0 \le s$. In this
case, our family of distributions is
\begin{equation}
  \label{eqn:compressed-sensing-P}
  \mc{P}_{X,s} \defeq \left\{Y \sim
  \normal(X \theta, \stddev^2 I_{n \times n})
  \mid \theta \in \R^d, \norm{\theta}_0 \le s \right\}.
\end{equation}
\citet[Theorem 1]{CandesDa13} derived a lower bound for this model
based on constructing a packing set with particular covariance
properties.  Here we derive such a result as a direct consequence of
Corollaries~\ref{corollary:fun-fano}
and~\ref{corollary:estimation-to-testing}.
\begin{corollary}
  \label{corollary:compressed-sensing}
  For the $s$-sparse compressed sensing
  family~\eqref{eqn:compressed-sensing-P}, there is a universal constant
  $c > 0$ such that
  \begin{equation}
    \label{eqn:tim-hortons}
    \minimax_n(\theta(\mc{P}_{X,s}), \ltwo{\cdot}^2)
    \ge c \frac{\stddev^2 s d \log(\frac{d}{s})}{\lfro{X}^2}.
  \end{equation}
\end{corollary}

For appropriate matrices $X$, the minimax risk is also upper bounded by the
right-hand side of equation~\eqref{eqn:tim-hortons}. Indeed, when the entries
of $X \in \R^{n \times d}$ are chosen as i.i.d.\ Gaussian random variables,
setting $\what{\theta}$ to be the minimizer of $\ltwo{X \theta - Y}^2 +
\lambda \lone{\theta}$ for appropriate $\lambda$ attains this
rate~\cite{ChandrasekaranRePaWi12}. For more general
random $X$, $\ell_1$-minimization techniques such as the Lasso or Dantzig
selector~\cite{CandesTa07,BickelRiTs09} achieve these rates with the quantity
$\log(d/s)$ replaced by $\log d$.  We now turn to the proof of the lower bound
in the claim~\eqref{eqn:tim-hortons}.

\begin{proof}
  As in the proof of Corollary~\ref{corollary:sparse-gaussian-location},
  we let the set $\packset = \{ \packval \in \{-1, 0, 1\}^d \, \mid \,
  \norm{\packval}_0 = s \}$,
  and we define $\theta_\packval = \epsilon \packval$ for some
  $\epsilon > 0$. The separation function is then lower bounded as
  $\delta(t) > \max\{\sqrt{t}, 1\} \: \epsilon$, and for
  $\packrv$ uniformly distributed on $\packset$,
  Corollary~\ref{corollary:estimation-to-testing} implies the lower
  bound
  \begin{align*}
    \minimax_n(\theta(\mc{P}_{X,s}), \ltwo{\cdot}^2) > \frac{ (t \vee 1) \;
      \epsilon^2 }{4} \left(1 - \frac{I(\packrv; Y) + \log 2}{\log
      \binom{d}{s} - \log \hoodbig_t}\right).
  \end{align*}
  We have $\cov(\packrv) = (s/d) I_{d \times d}$ for $V$ uniform on $\packrv$,
  so for $V, W$ independent and uniform on $\mc{V}$,
  \begin{align*}
    I(\packrv; Y) \le \frac{1}{|\packset|^2}\sum_{\packval \in \packset}
    \sum_{\altpackval \in \packset} \dkl{\normal(X \theta_\packval, \sigma^2
      I)}{ \normal(X \theta_\altpackval, \sigma^2 I)}
    & =
    \frac{\epsilon^2}{2 \stddev^2}
    \E\left[\ltwo{X \packrv - X W}^2\right] 
    = \frac{s \epsilon^2 \lfro{X}^2}{d}.
  \end{align*}
  By taking $t = \floor{s/4}$, we find that $\log
  \frac{|\packset|}{\hoodbig_t} \ge c s \log \frac{d}{s}$ for a numerical
  constant $c$ as was the case for
  Corollary~\ref{corollary:sparse-gaussian-location}. Combining the pieces
  yields the bound
  \begin{equation*}
    \minimax_n(\theta(\mc{P}_{s,d}),
    \ltwo{\cdot}^2) > \frac{ (\floor{s/4} \vee
      1) \; \epsilon^2}{4} \left(1 - \frac{s \epsilon^2 \lfro{X}^2 / (d
      \stddev^2)
      + \log 2}{c s \log \frac{d}{s}}\right),
  \end{equation*} 
  and setting $\epsilon^2 \asymp \stddev^2 d \log \frac{d}{s} / \lfro{X}^2$
  yields the claim~\eqref{eqn:tim-hortons}.
\end{proof}


\subsection{Applications of Proposition~\ref{proposition:continuous-fano}}
\label{sec:apply-continuous-fano}

The final set of applications of the results in this note concerns the direct
application of the continuous version of Fano's inequality,
Proposition~\ref{proposition:continuous-fano}.  In applications of the
continuous Fano inequality, we no longer require a reduction from the original
estimation problem to a discrete estimation problem as in
Corollary~\ref{corollary:estimation-to-testing} or the classical minimax Fano
bound~\eqref{eqn:classical-fano}:
Proposition~\ref{proposition:continuous-fano} provides an immediate
lower bound.

We first consider normal mean estimation, where
$\mc{P} = \{\normal(\theta, \stddev^2 I_{d \times d}) \mid \theta \in \R^d\}$
and our goal is to estimate the mean $\theta$ given $n$
i.i.d.\ observations $X_1^n$ from $\normal(\theta, \stddev^2
I)$. Proposition~\ref{proposition:continuous-fano} implies the following
corollary, which follows by an integration argument.
\begin{corollary}
  \label{corollary:normal-mean}
  For the $d$-dimensional normal location family with $d \ge 2$,
  \begin{equation*}
    \minimax_n(\theta(\mc{P}), \ltwo{\cdot}^2) \ge \frac{(d-1)^2 \log 2}{4 d^2}
    \, \frac{\stddev^2 d}{n}.
  \end{equation*}
\end{corollary}
\begin{proof}
  Let $V$ be uniform on the $\ell_2$-ball of radius $r$ centered at
  $0$, and conditioned on $V = v$, let $X_1^n = (X_1, \ldots, X_n)$ be
  an i.i.d.\ sample of size $n$ from the multivariate Gaussian
  distribution $\normal(v, \stddev^2 I_{d \times d})$.  Then for $0
  \le t \le r$, the volume ratio is given by $\vol(\packset) /
  \vol(\ball_2(t)) = (r/t)^d$, and
  Proposition~\ref{proposition:continuous-fano} implies that
  \begin{equation*}
    \P(\ltwos{\what{V} - V} \ge t)
    \ge 1 - \frac{I(V; X_1^n) + \log 2}{d \log \frac{r}{t}}.
  \end{equation*}
  Setting $r = 2t$, the denominator is equal to $d \log 2$, so
  it only remains to upper bound the mutual information.

  Since $V$ is uniform on the $\ell_2$-ball, we have $\Exs[X_i] = 0$.
  Consequently, from the independence of samples and the maximum entropy
  property of the Gaussian, we have
  \begin{align*}
    I(V; X_1^n) & = h(X_1^n) - h(X_1^n \mid V) \; \leq \; \frac{n}{2} \log
    \frac{\det( \cov(X_1))}{\det (\sigma^2 I_{d \times d})},
  \end{align*}
  using the fact that $X_1$ is Gaussian conditioned on $V$.  Since
  $V$ is uniform on the $\ell_2$-ball of radius $2t$, we have $\cov(X_1)
  = \sigma^2 I_{d \times d} + \Exs[V V^T] \preceq (\sigma^2 + 4 t^2)
  I_{\usedim \times \usedim}$.  Putting together the pieces yields
  \begin{equation}
    I(V; X_1^n) \leq \frac{n}{2} \log\left(1 + \frac{4 t^2}{\sigma^2}\right).
    \label{eqn:normal-location-information}
  \end{equation}

  We provide two arguments based on
  inequality~\eqref{eqn:normal-location-information}: the first is simpler,
  while the second provides sharper constants.
  For the first argument, the concavity of the $\log$ function
  immediately
  implies $I(V; X_1^n) \le 2 n t^2 / \stddev^2$, since 
  $\log(1 + x) \le x$.
  With this mutual
  information bound, we obtain
  \begin{equation*}
    \P(\ltwos{\what{V}(X_1^n) - V} \ge t)
    \ge 1 - \frac{\log 2}{d \log 2}
    - \frac{I(V; X_1^n)}{d \log 2}
    = \frac{d - 1}{d}
    - \frac{I(V; X_1^n)}{d \log 2}
    \ge \half - \frac{2 n t^2}{d \sigma^2 \log 2}.
  \end{equation*}
  Setting $t^2 = d \stddev^2 \log 2 / 4 n$ implies the estimation lower
  bound $\P(\ltwos{\what{V} - V}^2 \ge t^2) \ge 1/4$ directly
  from the volume bound in Proposition~\ref{proposition:continuous-fano}.

  To obtain the sharper inequality claimed in the corollary, we apply an
  integration argument. For any positive random variable $Y$, $\E[Y] =
  \int_0^\infty \P(Y \ge t) dt$, and integrating $\P(\ltwos{\what{V} - V} \ge
  t)$ bounds the expected error.
  As a consequence, we find from our information
  bound~\eqref{eqn:normal-location-information} and the continuous Fano
  inequality~\eqref{eqn:continuous-fano} that
  \begin{align}
    \lefteqn{
      \int_0^\infty \P(\ltwos{\what{V} - V} \ge \stddev \sqrt{t} / 2) dt
      \ge \int_0^\infty \hinge{\frac{d-1}{d}
        - \frac{n \log (1 + t)}{2 d \log 2}} dt} \nonumber \\
    & \qquad\qquad\qquad \quad ~ = \frac{n}{2d \log 2}
    \left[\exp\left(\frac{2 (d-1) \log 2}{n}\right)
      - 1 - \frac{2(d-1) \log 2}{n} \right]
    \ge \frac{(d-1)^2 \log 2}{d n}
    \label{eqn:annoying-integral}
  \end{align}
  (see Appendix~\ref{appendix:annoying-integral}
    for the computation of the integral).
  Rewriting inequality~\eqref{eqn:annoying-integral}, we have
  \begin{equation*}
    \frac{4}{\stddev^2}
    \E\left[\ltwos{\what{V} - V}^2\right]
    = \int_0^\infty \P(\ltwos{\what{V} - V}^2 \ge \stddev^2 t / 4) dt
    \ge \frac{(d - 1)^2 \log 2}{d} \, \frac{1}{n},
  \end{equation*}
  which is the claimed inequality of the corollary.
\end{proof}

As our second example application of the volume-based Fano inequality in
Proposition~\ref{proposition:continuous-fano}, we consider the standard
fixed-design linear regression model $Y = X\theta + \varepsilon$, where
$\varepsilon \in \R^n$ is i.i.d.\ $\normal(0, \stddev^2)$ and $X \in \R^{n
  \times d}$, where the goal is to estimate $\theta$. In this case,
our family of distributions is
\begin{equation}
  \label{eqn:linear-regression-model}
  \mc{P}_X \defeq \left\{\normal(X \theta, \stddev^2 I_{n \times n})
  \mid \theta \in \R^d \right\}
  = \left\{Y = X\theta + \varepsilon
  \mid \varepsilon \sim \normal(0, \stddev^2 I_{n \times n}),
  \theta \in \R^d \right\}.
\end{equation}
We make the simplifying assumption that $X \in \R^{n \times d}$ is of
full column rank, and that $d \ge 9$.  Letting $\gamma_{\max}(X)$
denote the maximum singular value of $X$, we have the following
corollary to Proposition~\ref{proposition:continuous-fano}:
\begin{corollary}
  For the standard linear regression model~\eqref{eqn:linear-regression-model},
  we have
  \begin{equation*}
    \minimax_n(\theta(\mc{P}_X), \ltwo{\cdot}^2)
    \ge \frac{1}{12} \cdot \frac{1}{\gamma_{\max}^2(X / \sqrt{n})} \cdot
    \frac{d \stddev^2}{n}.
  \end{equation*}
\end{corollary}
\begin{proof}
  Letting $\theta = V \in \R^d$ be uniform on the
  $\ell_2$-ball of radius $r$, we have
  \begin{align*}
    I(V; Y)
    & \le \int \dkl{\normal(Xv, \stddev^2 I_{n \times n})}{\normal(Xw,
      \stddev^2 I_{n \times n})} d\mu(v) d\mu(w) \\
    & = \frac{1}{2 \stddev^2}
    \E\left[\ltwo{X(V - W)}^2\right]
    = \frac{d}{d + 2} \frac{\tr(X^\top X)}{\stddev^2} r^2
    = \frac{d}{d + 2} \frac{\lfro{X}^2}{\stddev^2} r^2,
  \end{align*}
  where $\mu$ is uniform on the $\ell_2$-ball of radius $r$, as are
  the independent random vectors $V$ and $W$.  By setting $r = 2t$,
  Proposition~\ref{proposition:continuous-fano} implies that any
  estimator $\what{\theta}$ of $\theta$ must satisfy
  \begin{equation*}
    \P(\ltwos{\what{\theta} - \theta} \ge t)
    \ge \frac{d - 1}{d} - \frac{4 t^2 \lfro{X}^2}{\stddev^2 d(d + 2) \log 2}.
  \end{equation*}
  Applying an integration argument as in
  Corollary~\ref{corollary:normal-mean}, we use the identity $\int_0^\infty
  \hinge{c_1 - c_2 t} dt = c_1^2 / 2 c_2$ to show that for any estimator
  $\what{\theta}$ we have
  \begin{equation*}
    \E\left[\ltwos{\what{\theta}(X_1^n) - \theta}^2\right]
    \ge \int_0^\infty \hinge{\frac{d - 1}{d}
      - \frac{4t \lfro{X}^2}{\stddev^2 d (d + 2) \log 2}} dt
    = \frac{(d - 1)^2}{d^2}
    \cdot \frac{d(d + 2) \stddev^2 \log 2}{8 \lfro{X}^2}.
  \end{equation*}
  
  To slightly simplify the above bound, we note that $\lfro{X / \sqrt{n}}^2
  \le d \gamma_{\max}^2(X / \sqrt{n})$, which gives the minimax lower bound
  \begin{equation*}
    \minimax_n(\theta(\mc{P}), \ltwo{\cdot}^2)
    \ge \frac{\log 2}{8} \cdot \frac{(d - 1)^2(d + 2)}{d^3}
    \cdot \frac{1}{\gamma_{\max}^2(X / \sqrt{n})} \cdot
    \frac{d \stddev^2}{n}
    \ge \frac{1}{12}
    \cdot \frac{1}{\gamma_{\max}^2(X / \sqrt{n})} \cdot
    \frac{d \stddev^2}{n}
  \end{equation*}
  the last inequality following for $d \ge 9$.
\end{proof}


\section{Conclusion}

We have provided two quantitative Fano inequalities, in the forms of
Propositions~\ref{proposition:fun-fano} and~\ref{proposition:continuous-fano},
that allow direct proofs of several statistical minimax lower bounds.  It
would be interesting to see if these inequalities can be used in more
classical information-theoretic applications, either to give simpler proofs of
existing converse results or to prove new converse results. We hope to explore
these questions and other applications in future work.

\section*{Acknowledgments}
JCD was partially supported by a Facebook Fellowship, and both authors
were partially supported by ONR MURI grant N00014-11-1-0688 and
NSF grant CIF-31712-23800.
\appendix


\section{Proof of Corollary~\ref{corollary:estimation-to-testing}}
\label{sec:proof-estimation-to-test}
\newcommand{\sample}{X}

For any $\epsilon \geq 0$ and any estimator $\what{\theta}$ of $\theta$,
the non-decreasing nature of $\Phi$ implies that
\begin{align*}
\E[\Phi(\metric(\what{\theta}, \theta))] \ge \E\left[\Phi(\delta)
  \indic{\metric(\what{\theta}(\sample_1^n), \theta) \ge
    \epsilon}\right] & = \Phi(\epsilon) \;  \P
\Big(\metric(\what{\theta}(\sample_1^n), \theta) \ge
\frac{\epsilon}{2} \Big).
\end{align*}
Consequently, setting $\epsilon = \delta(t)$, the lower
bound~\eqref{EqnGeneralizedFano} in the corollary will follow from
Proposition~\ref{proposition:fun-fano} if we can show that
\begin{align}
\label{EqnCrestedJay}
\P \Big(\metric(\what{\theta}(\sample_1^n), \theta_\packrv) \geq
\frac{\delta(t)}{2} \Big) & \geq \P
\Big(\packmetric(\what{\packval}(\sample_1^n), \packrv) > t \Big).
\end{align}
In order to establish this claim, we define the testing function
\begin{equation}
  \what{\packval}(\sample_1^n) \defeq \argmin_{\packval \in \packset}
  \metric(\theta_\packval, \what{\theta}(\sample_1^n)).
  \label{eqn:test-selection}
\end{equation}
Now assume that $\metric(\what{\theta}, \theta_\packval) <
\frac{\delta(t)}{2}$.  Then for any $\altpackval$ with
$\packmetric(\altpackval, \packval) > t$, we have
\begin{equation*}
  \metric(\what{\theta}, \theta_\altpackval) \ge
  \metric(\theta_\packval, \theta_\altpackval) -
  \metric(\what{\theta}, \theta_\packval) > \delta(t) -
  \frac{\delta(t)}{2} \; = \; \frac{\delta(t)}{2}.
\end{equation*}
As a consequence, we have $\metric(\what{\theta}, \theta_\altpackval)
> \metric(\what{\theta}, \theta_\packval)$ for all $\altpackval$
with $\packmetric(\altpackval, \packval) > t$, and thus
the choice~\eqref{eqn:test-selection} must give a $\what{\packval}$ such
that $\packmetric(\what{\packval}, \packval) \le t$. In particular,
the event $\metric(\what{\theta}, \theta_\packval) < \delta(t) / 2$
implies that $\packmetric(\what{\packval}, \packval) \le t$,
and thus
by conditioning on $\packrv = \packval$, 
we find that for any $t \in \R$,
\begin{equation*}
  \P(\metric(\what{\theta}(\sample_1^n), \theta_\packrv) \ge \delta(t)
  / 2 \mid \packrv = \packval) \ge \P(\packmetric(\what{\packval},
  \packrv) > t \mid \packrv = \packval).
\end{equation*}
Averaging over all $\packval \in \packset$ and taking an infimum over
all tests $\what{\packval}$ yields the claim~\eqref{EqnCrestedJay}.


\section{Proof of Proposition~\ref{proposition:continuous-fano}}
\label{sec:proof-continuous-fano}

Throughout this proof, we use $A + B$ to denote Minkowski addition of
two sets $A$ and $B$ in $\R^d$.  The proof is based on a sequence of
partitions of the space $\packset$, each of which allows us to apply
Proposition~\ref{proposition:fun-fano}.  Let $\epsilon_n = 2^{-n}$,
and consider a grid of $\R^d$ into boxes whose vertices are at points
$\epsilon_n z$ for $z \in \Z$, where each block is of width
$\epsilon_n$, i.e.\ the boxes are given by translations of
$[-\epsilon_n/2, \epsilon_n/2]^d$.  Let $\mc{W}^{(n)}$ denote the
partition of $\packset$ into these boxes, where we abuse notation and
let $|\mc{W}^{(n)}|$ denote the number of boxes whose intersection
with $\packset$ has non-zero volume.  Assign an arbitrary indexing to
the blocks (and partial blocks) $\mc{W}^{(n)}$, and for a vector $v
\in \packset$, let $[v]_{\mc{W}^{(n)}}$ be an arbitrary (but fixed)
point $w$ of the box in $\mc{W}^{(n)}$ into which $v$ falls (break
ties at the boundaries arbitrarily; the boundaries have Lebesgue
measure zero so this is insignificant). In addition, let
\begin{equation*}
  \hoodsize_t(\mc{W}^{(n)})
  \defeq \sup_{v \in \packset}
  \left\{\card\{[v]_{\mc{W}^{(n)}} \mid v' \in \packset,
  \metric(v, v') \le t \right\}
\end{equation*}
be the maximum number of blocks touched in a radius $t$ of some
point $v \in \packset$.
See Figure~\ref{fig:partitionings} for a visual representation of
our construction.

\begin{figure}[t]
  \begin{center}
    \begin{tabular}{cc}
      \begin{overpic}[width=.45\columnwidth]
        {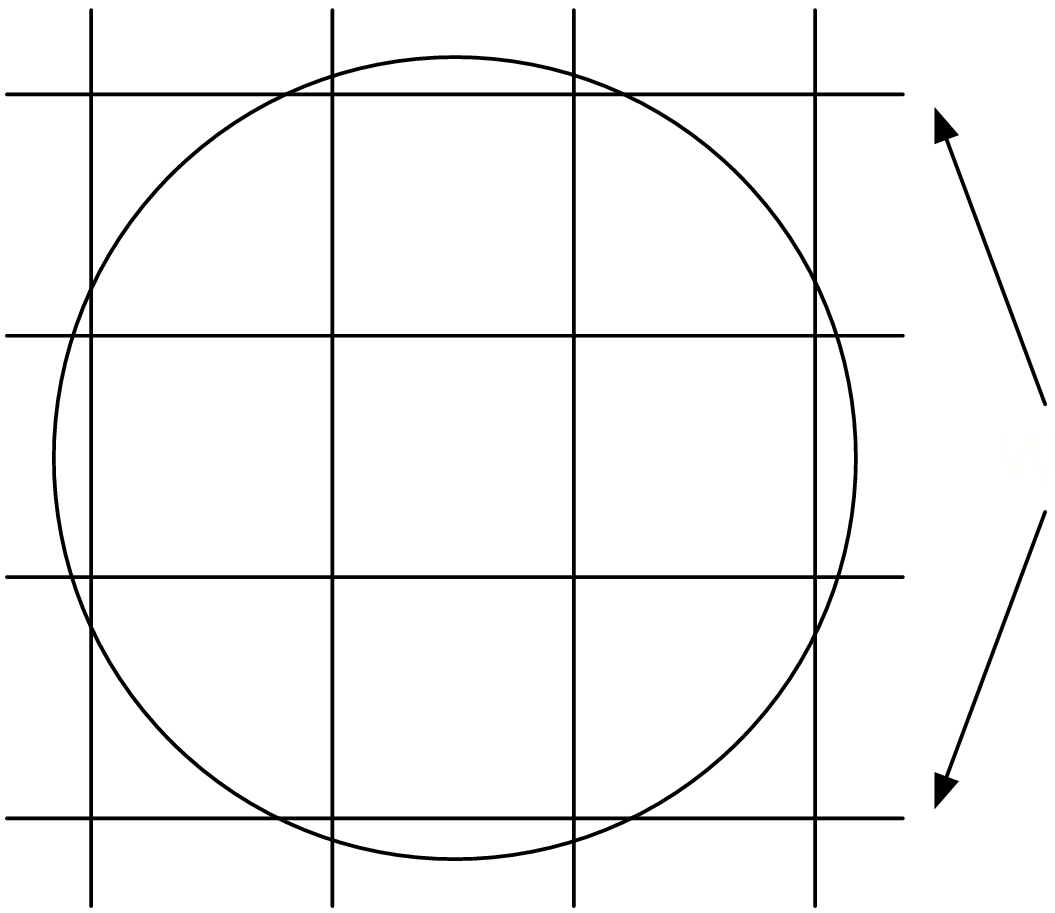}
        \put(32,40){\Large{$\packset$}}
        \put(84,36){\Large{$\mc{W}^{(n)}$}}
      \end{overpic} &
      \begin{overpic}[width=.45\columnwidth]
        {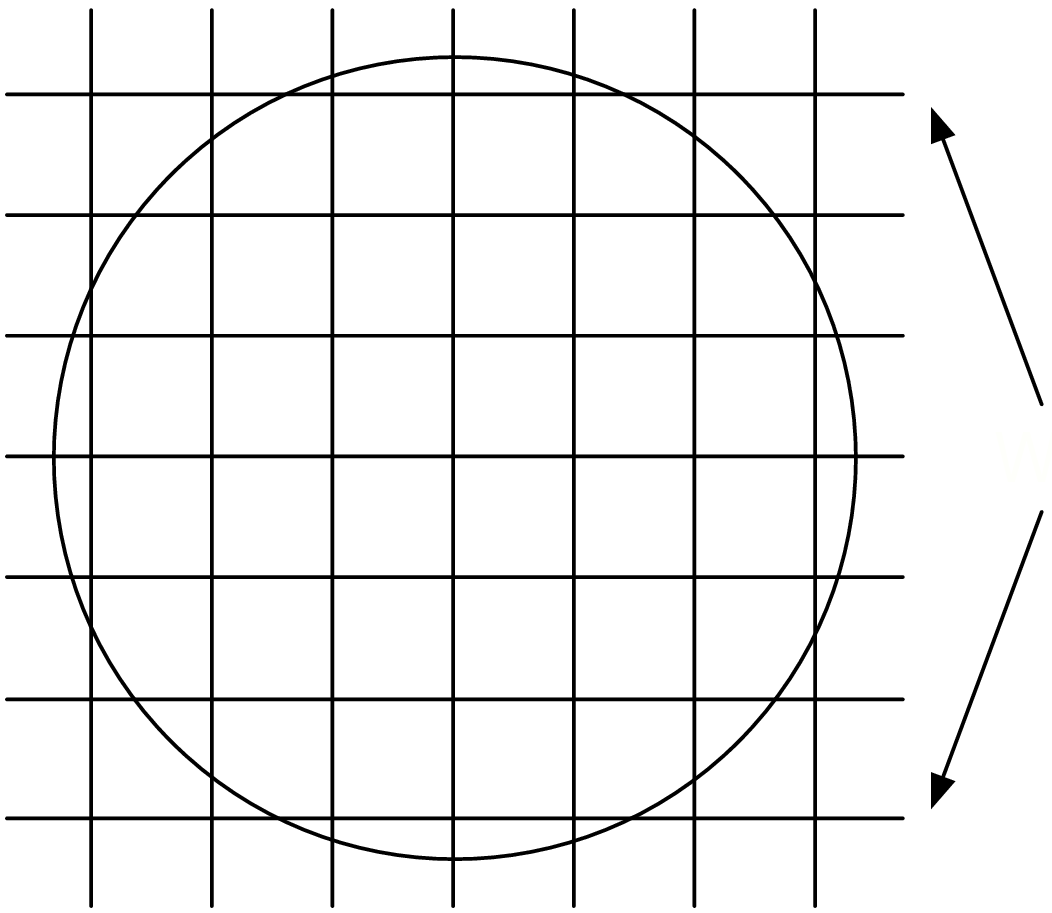}
        \put(32,40){\Large{$\packset$}}
        \put(84,36){\Large{$\mc{W}^{(n+1)}$}}
      \end{overpic} \\
      (a) & (b) \\
      \begin{overpic}[width=.45\columnwidth]
        {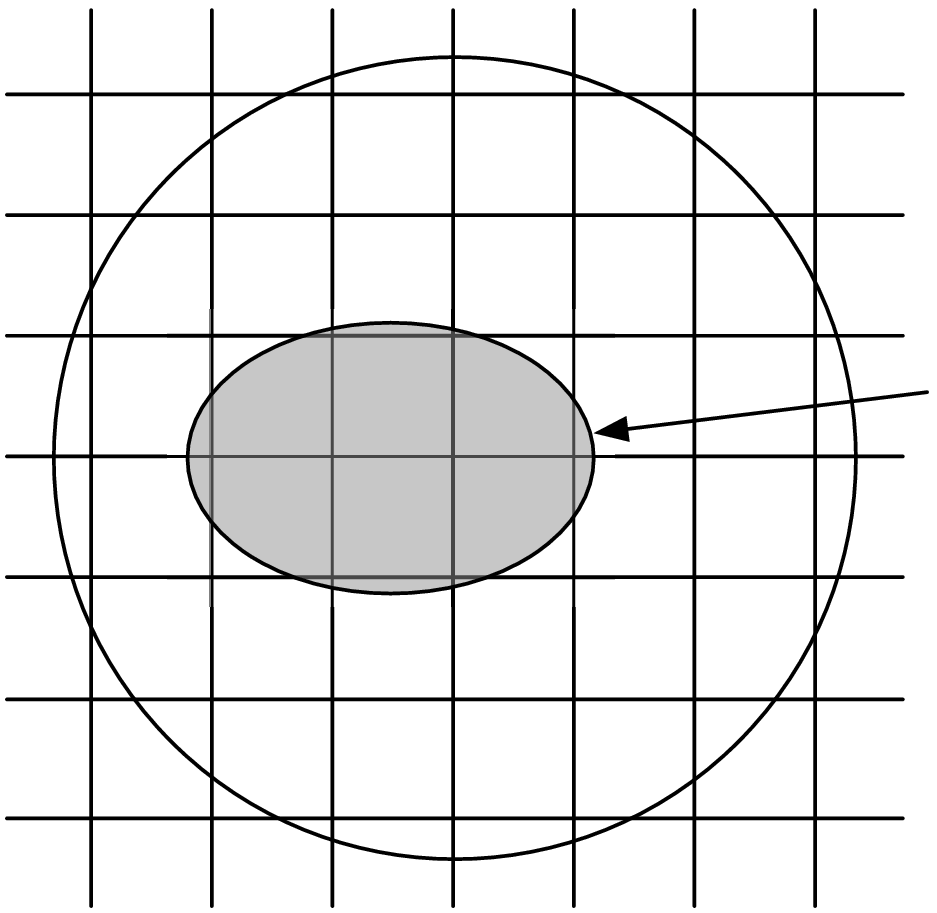}
        \put(79,42){\Large{$\ball_\metric(t, v)$}}
      \end{overpic} &
      \begin{overpic}[width=.484\columnwidth]
        {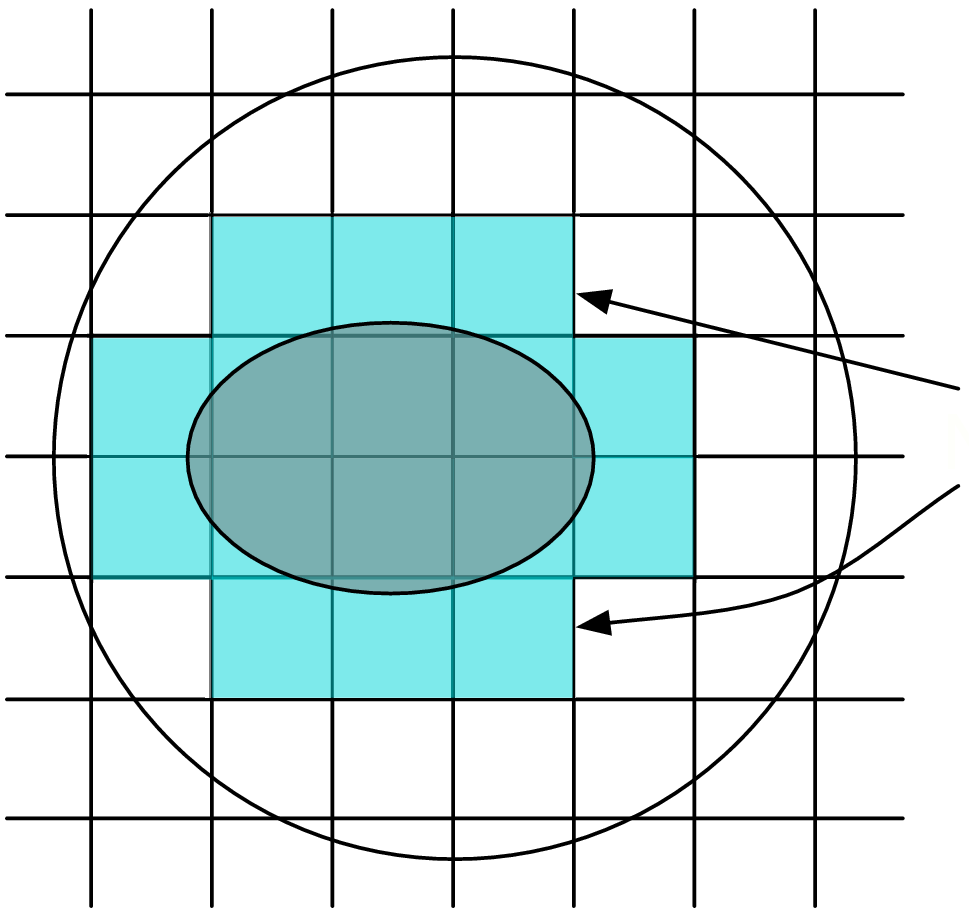}
        \put(74,36){\Large{$\hoodsize_t(\mc{W}^{(n)})$}}
      \end{overpic} \\
        (c) & (d)
    \end{tabular}
  \end{center}
  \caption{
    \label{fig:partitionings} Partitioning constructions in proof
    of Proposition~\ref{proposition:continuous-fano}.
    (a)~Partition of $\packset$ into $\mc{W}^{(n)}$.  (b)~Finer partition of
    $\packset$ into $\mc{W}^{(n+1)}$.  (c)~The ``ball'' $\ball_\metric(t,
    v)$. (d)~Sets counted in computing $\hoodsize_t(\mc{W}^{(n)})$.}
\end{figure}

With this notation, we may provide the proof.
For $V$ uniform on $\packset$, define the random variable $V^{(n)} =
[V]_{\mc{W}^{(n)}}$. By Proposition~\ref{proposition:fun-fano},
we have
\begin{align}
  \log 2 + \log \frac{|\mc{W}^{(n)}|}{N_t(\mc{W}^{(n)})}
  \P(\metric(\what{V}, V^{(n)}) > t)
  & \ge H(V^{(n)} \mid X) - \log N_t(\mc{W}^{(n)}) \nonumber \\
  & = H(V^{(n)} \mid X) - H(V^{(n)}) + H(V^{(n)}) - \log N_t(\mc{W}^{(n)}).
  \label{eqn:truncated-ineq}
\end{align}
Inspecting inequality~\eqref{eqn:truncated-ineq}, it suffices to show
that, for any $v \in \packset$, the following inequalities are valid:
\begin{subequations}
  \label{eqn:limits-to-prove}
  \begin{align}
    \liminf_n \P(\metric(\what{V}, V^{(n)}) > t) & \le
    \P(\metric(\what{V}, V) \ge t),
    \label{eqn:apply-prohorov} \\
    \liminf_n \log \frac{|\mc{W}^{(n)}|}{N_t(\mc{W}^{(n)})} & \le \log
    \frac{\vol(\packset)}{\vol(\ball_\metric(t, v) \cap \packset)},
    \qquad \mbox{and}
    \label{eqn:count-ws} \\
    \liminf_n \left[H(V^{(n)}) - \log N_t(\mc{W}^{(n)})\right] & \ge
    \log \frac{\vol(\packset)}{ \sup_{v \in \packset}
      \vol(\ball_\metric(t, v) \cap \packset)}.
    \label{eqn:lower-entropy}
  \end{align}
\end{subequations}
To see the sufficiency, suppose that the
conditions~\eqref{eqn:limits-to-prove} hold. Since $\mc{W}^{(n + 1)}$
partitions $\mc{W}^{(n)}$, we have that for all $n$,
\begin{equation*}
  I(V^{(n)}; X)
  \le I(V^{(n + 1)}; X) \le I(V; X)
\end{equation*}
(see~\cite[Chapter 5]{Gray90}).  As a consequence, we find that under
conditions~\eqref{eqn:limits-to-prove},
\begin{align*}
  \log 2 + \log \frac{\vol(\packset)}{\vol (B_\metric(t, v) \cap \packset)}
  \P(\metric(\what{V}, V) \ge t)
  & \stackrel{(i)}{\ge}
  \liminf_n\left\{
  \log 2 + \log \frac{|\mc{W}^{(n)}|}{N_t(\mc{W}^{(n)})}
  \P(\metric(\what{V}, V^{(n)}) > t)
  \right\} \\
  & \stackrel{(ii)}{\ge} \liminf_n \left\{-I(V^{(n)}; X) + H(V^{(n)}) - \log
  N_t(\mc{W}^{(n)})\right\} \\
  & \stackrel{(iii)}{\ge}
  -I(V; X) + \log \frac{\vol(\packset)}{
    \sup_{v \in \packset} \vol(\ball_\metric(t, v) \cap \packset)},
\end{align*}
where inequality~$(i)$ follows from inequalities~\eqref{eqn:count-ws}
and~\eqref{eqn:apply-prohorov}, $(ii)$ from the
consequence~\eqref{eqn:truncated-ineq} of
Proposition~\ref{proposition:fun-fano}, and $(iii)$ is a
consequence of inequality~\eqref{eqn:lower-entropy}.

We thus complete the proof of Proposition~\ref{proposition:continuous-fano} by
proving each of the inequalities~\eqref{eqn:limits-to-prove}.
\paragraph{Inequality \eqref{eqn:apply-prohorov}}
The simplest is inequality~\eqref{eqn:apply-prohorov}, which
follows by the Portmanteau theorem on convergence
in distribution (e.g.~\cite[Section 2]{Billingsley99}).
Since $V^{(n)} \to V$, the continuity of $\metric$ implies
$\metric(\what{V}, V^{(n)}) \to \metric(\what{V}, V)$. That the set
$\{z \in \R \mid z \ge t\}$ is closed implies that
\begin{equation*}
  \liminf_n \P(\metric(\what{V}, V^{(n)}) > t)
  \le \limsup_n \P(\metric(\what{V}, V^{(n)}) \ge t)
  \le \P(\metric(\what{V}, V) \ge t),
\end{equation*}
the last inequality following from the Portmanteau theorem.

\paragraph{Inequality \eqref{eqn:count-ws}}
For inequalities~\eqref{eqn:count-ws} and~\eqref{eqn:lower-entropy}, we
provide volume counting arguments.  We begin with two inequalities that will
prove useful.  Let $A$ be any set with finite surface area and volume and let
$\epsilon \ge 0$. Then, denoting $(d-1)$-dimensional Lebesgue measure by
$\lambda^{d-1}$,
\begin{equation}
  \begin{split}
    \vol(A \setminus (\partial A + [-\epsilon, \epsilon]^d))
    & \ge \vol(A) - \int_{\partial A}
    \vol([-\epsilon, \epsilon]^d) d\lambda^{d-1}
    = \vol(A) - (2\epsilon)^d \vol(\partial A) \\
    \vol(A + [-\epsilon, \epsilon]^d)
    & \le \vol(A) + \int_{\partial A}
    \vol([-\epsilon, \epsilon]^d) d\lambda^{d-1}
    = \vol(A) + (2 \epsilon)^d \vol(\partial A).
    \label{eqn:volume-surface-epsilon}
  \end{split}
\end{equation}

Now we turn to our counting arguments.
For the set $\mc{W}^{(n)}$, we have
\begin{equation}
  \label{eqn:count-w-size}
  \frac{1}{\epsilon_n^d}
  \vol(\packset \setminus (\partial \packset + [-\epsilon_n, \epsilon_n]^d))
  \le
  |\mc{W}^{(n)}|
  \le \frac{1}{\epsilon_n^d}
  \vol(\packset + [-\epsilon_n, \epsilon_n]^d).
\end{equation}
By construction of the set $\mc{W}^{(n)}$, we additionally have
for any $v \in \packset$ that
\begin{equation}
  \begin{split}
    \frac{1}{\epsilon_n^d}
    \vol(\ball_\metric(t, v) \cap \packset \setminus
    (\partial (\ball_\metric(t) \cap \packset)
    + [-\epsilon_n, \epsilon_n]^d))
    & \le
    \hoodsize_t(\mc{W}^{(n)}) \\
    & \le \sup_{v \in \packset} \frac{1}{\epsilon_n^d}
    \vol(\ball_\metric(t, v) \cap \packset + [-\epsilon_n, \epsilon_n]^d).
  \end{split}
  \label{eqn:count-hoodsize}
\end{equation}
As a consequence of the upper and lower volume
bounds~\eqref{eqn:count-w-size} and~\eqref{eqn:count-hoodsize},
we may prove inequality~\eqref{eqn:count-ws}:
for any $v \in \packset$ we have
\begin{align*}
  \liminf_n \log \frac{|\mc{W}^{(n)}|}{\hoodsize_t(\mc{W}^{(n)})}
  & \le \liminf_n
  \log \frac{\vol(\packset + [-\epsilon_n, \epsilon_n]^d)}{
    \vol(\ball_\metric(t, v) \cap \packset
    \setminus (\partial (\ball_\metric(t, v) \cap \packset)
    + [-\epsilon_n, \epsilon_n]^d))} \\
  & = \log \frac{\vol(\packset)}{\vol(\ball_\metric(t, v) \cap \packset)}
\end{align*}
by the assumed finiteness of $\vol(\partial \packset)$ (recall the
inequalities~\eqref{eqn:volume-surface-epsilon}). This gives
inequality~\eqref{eqn:count-ws}.

\paragraph{Inequality \eqref{eqn:lower-entropy}:}

It remains to prove inequality~\eqref{eqn:lower-entropy}.  For this, we note
that since $V$ is uniform on $\packset$, there are at least
\begin{equation*}
  \frac{\vol(\packset \setminus (\partial \packset
    + [-\epsilon_n, \epsilon_n]^d))}{\epsilon_n^d}
\end{equation*}
values of $V^{(n)}$ that occur with probability $\epsilon_n^d
/ \vol(\packset)$---the probability that $V$ falls in a partition
of $\mc{W}^{(n)}$ contained completely within $\packset$.
As a consequence, we obtain
\begin{equation*}
  H(V^{(n)})
  \ge \frac{\epsilon_n^d \vol(\packset \setminus (\partial\packset
    + [-\epsilon_n, \epsilon_n]^d))}{
    \epsilon_n^d \vol(\packset)}
  \log \frac{\vol(\packset)}{\epsilon_n^d}
  = \frac{\vol(\packset \setminus (\partial\packset
    + [-\epsilon_n, \epsilon_n]^d))}{\vol(\packset)}
  \log \frac{\vol(\packset)}{\epsilon_n^d},
\end{equation*}
which implies (with the counting estimate~\eqref{eqn:count-hoodsize}) that
\begin{align}
  \lefteqn{H(V^{(n)}) - \log N_t(\mc{W}^{(n)})} \nonumber \\
  & \ge
  \frac{\vol(\packset \setminus (\partial\packset
    + [-\epsilon_n, \epsilon_n]^d))}{\vol(\packset)}
  \log \frac{\vol(\packset)}{\epsilon_n^d}
  - \sup_{v \in \packset} \log \frac{
    \vol(\ball_\metric(t, v) \cap \packset + [-\epsilon_n, \epsilon_n]^d)}{
    \epsilon_n^d} \nonumber \\
  & =
  \left(\frac{\vol(\packset \setminus (\partial\packset
    + [-\epsilon_n, \epsilon_n]^d))}{\vol(\packset)} - 1\right)\log
  \frac{\vol(\packset)}{\epsilon_n^d}
  + \log \frac{\vol(\packset)}{\sup_{v \in \packset}
    \vol(\ball_\metric(t, v) \cap \packset + [-\epsilon_n, \epsilon_n]^d)}.
  \label{eqn:entropy-hood-diff}
\end{align}
Now we use the regularity assumption on $\packset$. By
inequalities~\eqref{eqn:volume-surface-epsilon},
the first ratio in the preceding display is
$\vol(\packset \setminus (\partial \packset +
[-\epsilon_n, \epsilon_n]^d)) / \vol(\packset)
= 1 + o(|\log \epsilon_n|^{-1})$,
so
\begin{equation*}
  \liminf_n \left(\frac{\vol(\packset \setminus (\partial\packset
    + [-\epsilon_n, \epsilon_n]^d))}{\vol(\packset)} - 1\right)\log
  \frac{\vol(\packset)}{\epsilon_n^d}
  = 0.
\end{equation*}
The second term in the display~\eqref{eqn:entropy-hood-diff} similarly
satisfies (again by the inequalities~\eqref{eqn:volume-surface-epsilon})
\begin{align*}
  \sup_{v \in \packset} \vol\left(\ball_\metric(t, v) \cap \packset\right)
  & \le \sup_{v \in \packset} \vol\left(\ball_\metric(t, v) \cap \packset
  + [-\epsilon_n, \epsilon_n]^d\right) \\
  & \le \sup_{v \in \packset}
  \left\{\vol\left(\ball_\metric(t, v) \cap \packset\right)
  + 2^d \epsilon_n^d \vol(\partial(\ball_\metric(t, v) \cap \packset))\right\}.
\end{align*}
By the regularity assumptions on $\packset$ and $\metric$, the second
term in expression~\eqref{eqn:entropy-hood-diff} thus has limit
$\log\frac{\vol(\packset) / \sup_{v \in \packset}}{ \vol(\ball_\metric(t, v)
  \cap \packset)}$ as $n \to \infty$.  This, in turn, completes the proof of
inequality~\eqref{eqn:lower-entropy}.


\section{Proof of inequality~\eqref{eqn:annoying-integral}}
\label{appendix:annoying-integral}

We wish to compute the integral $\int_0^\infty \hinge{\frac{d-1}{d} -
  \frac{n \log (1 + t)}{2 d \log 2}} dt$.  We first use the fact that
the indefinite integral of $\log(1 + x)$ is $(x + 1) \log(x + 1) - x$
and that for $t \ge 0$ we have
\begin{equation*}
\frac{n \log(1 + t)}{2 d \log 2} \le \frac{d - 1}{d} ~~ \mbox{if and
  only if} ~~ t \le t_{\max} \defeq \exp\left(\frac{2(d-1) \log
  2}{n}\right) - 1.
\end{equation*}
As a consequence, we find that
\begin{align*}
\lefteqn{\int_0^\infty \hinge{\frac{d-1}{d} - \frac{n \log (1 + t)}{2
      d \log 2}} dt = t_{\max}\frac{d-1}{d} - \frac{n}{2d \log 2}
  \left[(t_{\max} + 1) \log(t_{\max} + 1) - t_{\max}\right]} \\ 
& = \frac{d-1}{d} \left[\exp\left(\frac{2(d-1) \log 2}{n}\right) - 1 -
  \exp\left(\frac{2 (d - 1) \log 2}{n}\right)\right] + \frac{n}{2 d
  \log 2}\left[\exp\left(\frac{2 (d-1) \log 2}{n}\right) - 1 \right]
\\ 
& = \frac{n}{2d \log 2} \left[\exp\left(\frac{2 (d-1) \log
    2}{n}\right) - 1 - \frac{2(d-1) \log 2}{n} \right] \ge
\frac{(d-1)^2 \log 2}{d n},
\end{align*}
where the final inequality follows by the Taylor expansion of
$\exp(\cdot)$, since $\exp(x) \ge 1 + x + x^2 / 2$.


\bibliographystyle{abbrvnat}
\bibliography{bib}

\end{document}